\documentclass[11pt,a4paper]{article}
\usepackage{indentfirst,mathrsfs}
\usepackage{amsfonts,amsmath,amssymb,amsthm}
\usepackage{latexsym,amscd}
\usepackage{amsbsy}
\usepackage[all]{xy}
\usepackage{graphicx}
\usepackage{threeparttable}
\usepackage[usenames,dvipsnames,svgnames,table]{xcolor}
\usepackage{cite}

\newtheorem{thm}{Theorem}
\newtheorem{lem}{Lemma}

\newtheorem{example}{Example}

\begin{document}
\date{}
\title{AMDS Symbol-Pair Codes from Repeated-Root Cyclic Codes}
\maketitle
\begin{center}
\author{\large Junru Ma 
\footnote{Corresponding author
\par
Junru Ma is with Faculty of Mathematics and Statistics, Hubei Key Laboratory of Applied Mathematics, Hubei University, Wuhan 430062, Hubei, China.
\par
E-mails: majunru@hubu.edu.cn(J.Ma).
}}
\end{center}

\begin{quote}
{\small {\bf Abstract:} \ \
Symbol-pair codes are proposed to guard against pair-errors in symbol-pair read channels.
The minimum symbol-pair distance is of significance in determining the error-correcting capability of a symbol-pair code.
One of the central themes in symbol-pair coding theory is the constructions of symbol-pair codes with largest possible minimum pair distance.
Maximum distance separable (MDS) and  almost maximum distance separable (AMDS) symbol-pair codes are optimal and sub-optimal regarding to the Singleton bound, respectively.
In this paper, six new classes of AMDS symbol-pair codes are explicitly constructed through repeated-root cyclic codes and one class of such codes has unbounded lengths and the minimum symbol-pair distance can reach $13$.
}


\end{quote}

\section{Introduction}

Due to the development of modern high-density data storage systems, the reading process may be lower than that of the process used to store the data.
Symbol-pair code was proposed by Cassuto and Blaum to protect against pair-errors over symbol-pair channels in \cite{CB1}.
Cassuto and Blaum studied symbol-pair codes on pair-error correctability conditions, code construction, decoding methods and asymptopic bounds in \cite{CB1,CB2}.
Later, Cassuto and Litsyn \cite{CL} established that codes for pair-errors exist with strictly higher rates compared to codes for the Hamming metric with the same relative distance.
From now on,
researchers further investigated symbol-pair codes, including the construction of symbol-pair codes \cite{CJKWY,CKW,CLL,DGZZZ,DNS,EGY,KZL,KZZLC,LG,ML1}, some decoding algorithms of symbol-pair codes \cite{HMH,LXY,MHT,THM,YBS} and the symbol-pair weight distribution of some linear codes \cite{DNSS1,DNSS2,DWLS,ML2,SZW}.


The minimum symbol-pair distance plays an important role in determining the error-correcting capability of a symbol-pair code.
It is shown in  \cite{CB1} that a code $\mathcal{C}$ with minimum symbol-pair distance $d_{p}$ can correct up to $\lfloor\frac{d_{p}-1}{2}\rfloor$ symbol-pair errors.
In 2012, Chee et al. \cite{CKW} derived a Singleton bound on symbol-pair codes.
Compared with classical error-correcting codes, the symbol-pair codes achieving the Singleton bound are called MDS symbol-pair codes.
AMDS symbol-pair codes are sub-optimum regarding to the Singleton bound, which yields that it is valuable to construct AMDS symbol-pair codes.
Compared with other linear codes, cyclic codes are easier to realize coding and decoding. Cyclic codes are one of the powerful tools to construct symbol-pair codes with largest minimum symbol-pair distance.
In this paper, six new classes of AMDS symbol-pair codes with $d_p\in\{4,\,6,\,7,\,8,\,12,\,13\}$ on repeated-root cyclic codes are derived and one class of such codes has unbounded lengths.

The rest of this paper is organized as follows. In Section $2$, we introduce some basic notations and results on symbol-pair codes and constacyclic codes.
By means of repeated-root cyclic codes, we investigate AMDS symbol-pair codes in Section $3$.
In Section $4$, we make some conclusions.

\section{Preliminaries}

In this section, we review some basic notations and results on symbol-pair codes and constacyclic codes, which will be applied in the sequel.

\subsection{Symbol-pair Codes}

Let $q=p^{m}$ and $\mathbb{F}_{q}$ denote the finite field with $q$ elements, where $p$ is a prime and $m$ is a positive integer.
Throughout this paper, let $\star$ be an element in $\mathbb{F}_{q}^{*}$ and $\mathbf{0}$ denotes the all-zero vector.
Let $n$ be a positive integer.
From now on, we always take the subscripts modulo $n$.
For any vector $\mathbf{x}=\left(x_{0}, x_{1}, \cdots, x_{n-1}\right)$ in $\mathbb{F}_{q}^{n}$, the symbol-pair read vector of $\mathbf{x}$ is defined by
\begin{equation*}
  \pi\left(\mathbf{x}\right)=\left(\left(x_{0},\,x_{1}\right),
  \,\left(x_{1},\,x_{2}\right),\cdots,\left(x_{n-2},\,x_{n-1}\right),
  \,\left(x_{n-1},\,x_{0}\right)\right).
\end{equation*}
Denote by $\mathbb{Z}_{n}$ the residue class ring $\mathbb{Z}/n\mathbb{Z}$.
The {\it symbol-pair weight} of $\mathbf{x}$ is
\begin{equation*}
  w_{p}\left(\mathbf{x}\right)=\left|\left\{i\in \mathbb{Z}_{n}\,\big|\,\left(x_{i},\,x_{i+1}\right)\neq \left(0,\,0\right)\right\}\right|.
\end{equation*}
For any $\mathbf{x},\,\mathbf{y}\in \mathbb{F}_{q}^{n}$, the {\it symbol-pair distance} between $\mathbf{x}$ and $\mathbf{y}$ is
\begin{equation*}
  d_{p}\left(\mathbf{x},\,\mathbf{y}\right)=\left|\left\{i\in \mathbb{Z}_{n}\,\big|\,\left(x_{i},\,x_{i+1}\right)\neq \left(y_{i},\,y_{i+1}\right)\right\}\right|.
\end{equation*}
A code $\mathcal{C}$ is said to have {\it minimum symbol-pair distance} $d_{p}$ if
\begin{equation*}
  d_{p}={\rm min}\left\{d_{p}\left(\mathbf{x},\,\mathbf{y}\right)\,\big|\, \mathbf{x},\,\mathbf{y}\in \mathcal{C},\mathbf{x}\neq\mathbf{y}\right\}.
\end{equation*}
Elements of $\mathcal{C}$ are called {\it codewords} in $\mathcal{C}$.
It is shown in \cite{CB1,CB2} that for any $0<d_H(\mathcal{C})<n$,
\begin{equation}\label{eqdhdp}
  d_H(\mathcal{C})+1\leq d_{p}\left(\mathcal{C}\right)\leq 2\cdot d_H(\mathcal{C}).
\end{equation}

Similar to classical error-correcting codes, the size of a symbol-pair code satisfies the following Singleton bound.

\begin{lem}{\rm(\,\cite{CJKWY}\,)}\label{Singleton}
Let $q\geq 2$ and $2\leq d_{p}\leq n$.
If $\mathcal{C}$ is a symbol-pair code with length $n$ and minimum symbol-pair distance $d_{p}$, then $\left|\mathcal{C}\right|\leq q^{n-d_{p}+2}$.
\end{lem}

The symbol-pair code achieving the Singleton bound is called a {\it maximum distance separable} (MDS) symbol-pair code.
For a linear code of length $n$, dimension $k$ and minimum symbol-pair distance $d_p$, if $d_p=n-k+1$, then it is called an {\it almost maximum distance separable} (AMDS) symbol-pair code.

\subsection{Constacyclic Codes}

In this subsection, we review some basic concepts of constacyclic codes.
For any $\eta\in\mathbb{F}_{q}^{*}$, the $\eta$-constacyclic shift $\tau_{\eta}$ on $\mathbb{F}_{q}^n$ is defined as
\begin{equation*}
  \tau_{\eta}\left(x_0,\,x_1,\cdots,x_{n-1}\right)
  =\left(\eta\,x_{n-1},\,x_0,\cdots,x_{n-2}\right).
\end{equation*}
A linear code $\mathcal{C}$ is an {\it$\eta$-constacyclic code} if $\tau_{\eta}\left(\mathbf{c}\right)\in\mathcal{C}$ for any $\mathbf{c}\in\mathcal{C}$.
An $\eta$-constacyclic code is a {\it cyclic code} if $\eta=1$.
Each codeword $\mathbf{c}=\left(c_0,\,c_1,\cdots,c_{n-1}\right)\in\mathcal{C}$ can be identified with a polynomial
\begin{equation*}
  c(x)=c_0+c_1\,x+\cdots+c_{n-1}\,x^{n-1}.
\end{equation*}
In this paper, we always regard the codeword $\mathbf{c}$ in $\mathcal{C}$ as the corresponding polynomial $c(x)$.
There is a unique monic polynomial $g(x)\in\mathbb{F}_{q}[x]$ with $g(x)\,|\,\left(x^n-\eta\right)$ and
\begin{equation*}
  \mathcal{C}=\left\langle g(x)\right\rangle=\left\{f(x)\,g(x)\,\left({\rm mod}\,x^n-\eta\right)\,\big|\,f(x)\in\mathbb{F}_{q}\left[x\right]\right\}.
\end{equation*}
We refer $g(x)$ as the {\it generator polynomial} of $\mathcal{C}$ and the dimension of $\mathcal{C}$ is $n-{\rm deg}\left(g(x)\right)$.


Let $\mathcal{C}=\left\langle g(x)\right\rangle$ be a repeated-root cyclic code of length $lp^e$ over $\mathbb{F}_{q}$ with ${\rm gcd}\left(l,\,p\right)=1$ and
\begin{equation*}
  g(x)=\prod_{i=1}^{r}m_{i}(x)^{e_i}
\end{equation*}
the factorization of $g(x)$ into distinct monic irreducible polynomials $m_i(x)\in\mathbb{F}_q\left[x\right]$ of multiplicity $e_i$.
For any $0\leq t\leq p^e-1$, we denote $\overline{\mathcal{C}}_{t}$ the simple-root cyclic code of length $l$ over $\mathbb{F}_q$ with generator polynomial
\begin{equation*}
  \overline{g}_t(x)=\prod_{1\leq i\leq r,\,e_i>t}m_i(x).
\end{equation*}
If this product is $x^{l}-1$, then $\overline{\mathcal{C}}_{t}$ contains only the all-zero codeword and we set $d_H(\overline{\mathcal{C}}_{t})=\infty$.
If all $e_i(1\leq i\leq s)$ satisfy $e_i\leq t$, then $\overline{g}_t(x)=d_H(\overline{\mathcal{C}}_{t})=1$.

The following lemma indicates that the minimum Hamming distance of $\mathcal{C}$ can be obtained by the minimum Hamming distance $d_{H}(\overline{\mathcal{C}}_{t})$, which is useful to derive the minimum Hamming distance of codes in Sections $3$.

\begin{lem}{\rm(\,\cite{CMSS}\,)}\label{lemdistance}
Let $\mathcal{C}$ be a repeated-root cyclic code of length $lp^e$ over $\mathbb{F}_{q}$, where $l$ and $e$ are positive integers with ${\rm gcd}\left(l,\,p\right)=1$.
Then
\begin{equation}\label{eqdistance}
  d_{H}(\mathcal{C})={\rm min}\left\{P_{t}\cdot d_{H}\left(\overline{\mathcal{C}}_{t}\right)\,\big|\,0\leq t\leq p^e-1\right\}
\end{equation}
where
\begin{equation}\label{eqpt}
  P_{t}=w_{H}\left((x-1)^t\right)=\prod_{i}\left(t_{i}+1\right)
\end{equation}
with $t_{i}$'s being the coefficients of the $p$-adic expansion of $t$.
\end{lem}

In this paper, we will employ repeated-root cyclic codes to construct new AMDS symbol-pair codes.
The following lemmas will be frequently applied in our later proof.


\begin{lem}{\rm(\,\cite{ML1}\,)}\label{lemdH}
Let $\mathcal{C}$ be a repeated-root cyclic code of length $lp^e$ over $\mathbb{F}_{q}$ and $c(x)=(x^l-1)^{t}\,v(x)$ an arbitrary nonzero codeword in $\mathcal{C}$, where $l$ and $e$ are positive integers with ${\rm gcd}\left(l,\,p\right)=1$, $0\leq t\leq p-1$ and $\left(x^l-1\right)\nmid v(x)$.
If $w_{H}\left(c(x)\right)=d_{H}(\mathcal{C})$, then
\begin{equation*}
  w_{H}\left(c(x)\right)=P_{t}\cdot N_{v}.
\end{equation*}
Here $P_t$ is defined as (\ref{eqpt}) in Lemma \ref{lemdistance} and $N_{v}=w_{H}\left(v(x)\,{\rm mod}\,(x^l-1)\right)$.
\end{lem}

\begin{lem}{\rm(\,\cite{ML1}\,)}\label{lemML}
Let $\mathcal{C}$ be a repeated-root cyclic code over $\mathbb{F}_{q}$ with length $3p$ and $g(x)=\left(x-1\right)^{5}\left(x-\omega\right)^3\left(x-\omega^2\right)^2$ the generator polynomial of $\mathcal{C}$, where $3\,|\,(p-1)$ and $\omega$ is a primitive third root of unity in $\mathbb{F}_{p}$. Then
$\mathcal{C}$ is an MDS $(3p,\,12)$ symbol-pair code.
\end{lem}

\begin{lem}{\rm(\,\cite{CLL}\,)}\label{lemMDS}
Let $\mathcal{C}$ be an $[n,\,k,\,d_{H}]$ constacyclic code over $\mathbb{F}_{q}$ with $2\leq d_{H}<n$.
Then we have $d_{p}(\mathcal{C})\geq d_{H}+2$ if and only if $\mathcal{C}$ is not an MDS code, i.e., $k<n-d_{H}+1$.
\end{lem}

\begin{lem}{\rm(\,\cite{CLL}\,)}\label{lemdh3}
Let $\mathcal{C}$ be an $\left[lp^e,\,k,\,d_{H}\right]$ repeated-root cyclic code over $\mathbb{F}_{q}$ and $g(x)$ the generator polynomial of $\mathcal{C}$, where ${\rm gcd}\left(l,p\right)=1$ and $l,\,e>1$.
If $d_{H}(\mathcal{C})$ is prime and one of the following two conditions is satisfied

(1) $l<d_{H}(\mathcal{C})<lp^e-k$;

(2) $x^l-1$ is a divisor of $g(x)$ and $2<d_{H}(\mathcal{C})<lp^e-k$,
\\[2mm]
then $d_{p}\left(\mathcal{C}\right)\geq d_{H}(\mathcal{C})+3$.
\end{lem}

\section{Constructions of AMDS Symbol-Pair Codes}

By utilizing repeated-root cyclic codes, we construct six new classes of AMDS symbol-pair codes with $d_p\in\{4,\,6,\,7,\,8,\,12,\,13\}$ in this section.
Among these codes, one class has unbounded lengths.

Now we begin with a class of AMDS symbol-pair codes with unbounded lengths.

\begin{thm}\label{thmAMDSlps4q}
Let $q$ be a power of prime $p$, $s\geq1$, $l$ a positive integer with $l\,|\left(q^2-1\right)$ and $l\geq \left(q+1\right)$.
Then there exists an AMDS $\left(lp^s,\,4\right)_{q}$ symbol-pair code.
\end{thm}

\begin{proof}
Let $\mathcal{C}$ be a repeated-root cyclic code of length $lp^s$ over $\mathbb{F}_q$ with generator polynomial
\begin{equation*}
  g(x)=\left(x-1\right)\left(x-\delta\right)\left(x-\delta^{q}\right)
\end{equation*}
where $\delta$ is a primitive $l$-th root of unity in $\mathbb{F}_{q^2}$.
By Lemma \ref{lemdistance}, one can derive that the parameter of $\mathcal{C}$ is $\left[lp^s,\,lp^s-3,\,2\right]$.
Since $\mathcal{C}$ is not MDS, Lemma \ref{lemMDS} yields that $d_{p}\left(\mathcal{C}\right)\geq4$.
Note that by inequality (\ref{eqdhdp}), one can obtain
\begin{equation*}
  d_{H}(\mathcal{C})+1\leq d_{p}\left(\mathcal{C}\right)\leq 2\cdot d_{H}(\mathcal{C})=4.
\end{equation*}
It follows that $d_{p}(\mathcal{C})=4$, which induces that $\mathcal{C}$ is an AMDS $\left(lp^s,\,4\right)_{q}$ symbol-pair code.
This concludes the desired result.
\end{proof}

In what follows, according to Lemma \ref{lemdh3}, we derive a class of AMDS symbol-pair codes with $n=4p$ and $d_{p}=6$.

\begin{thm}\label{thmAMDS4p6p}
Let $p$ be an odd prime.
Then there exists an AMDS $\left(4p,\,6\right)_{p}$ symbol-pair code.
\end{thm}

\begin{proof}
Let $\mathcal{C}$ be a repeated-root cyclic code of length $4p$ over $\mathbb{F}_p$ with generator polynomial
\begin{equation*}
g(x)=\left(x-1\right)\left(x^4-1\right).
\end{equation*}
It follows from Lemma \ref{lemdistance} that $\mathcal{C}$ is a $\left[4p,\,4p-5,\,3\right]$ cyclic code.
By Lemma \ref{lemdh3}, one can immediately get $d_{p}\left(\mathcal{C}\right)\geq6$
since $\left(x^4-1\right)\,\big|\,g(x)$ and $2<3=d_{H}(\mathcal{C})<4p-(4p-5)=5$.
Note that
\begin{equation*}
  d_{p}(\mathcal{C})\leq 2\cdot d_{H}(\mathcal{C})=6.
\end{equation*}
Hence $\mathcal{C}$ is an AMDS $\left(4p,\,6\right)_{p}$ symbol-pair code.
This completes the proof.
\end{proof}

In the sequel, four constructions of AMDS symbol-pair codes with length $3p$ are provided.
The first one posses the minimum pair-distance $7$ for any prime $p\geq5$.

\begin{thm}\label{thmAMDS3p7p}
Let $p\geq5$ be an odd prime.
Then there exists an AMDS $\left(3p,\,7\right)_{p}$ symbol-pair code.
\end{thm}

\begin{proof}
Let $\mathcal{C}$ be a repeated-root cyclic code of length $3p$ over $\mathbb{F}_p$ with generator polynomial
\begin{equation*}
  g(x)=\left(x-1\right)^{4}\left(x^{2}+x+1\right).
\end{equation*}
According to Lemma \ref{lemdistance}, one can obtain that $\mathcal{C}$ is a $\left[3p,\,3p-6,\,4\right]$ cyclic code.
Since $\mathcal{C}$ is not MDS, by Lemma \ref{lemMDS}, one can deduce that $d_{p}\left(\mathcal{C}\right)\geq6$.

Suppose that there exists a codeword $c(x)$ in $\mathcal{C}$ with Hamming weight $5$ such that $c(x)$ has five consecutive nonzero entries.
Without loss of generality, we assume that
\begin{equation*}
  c(x)=c_{0}+c_{1}\,x+c_{2}\,x^{2}+c_{3}\,x^{3}+c_{4}\,x^{4}
\end{equation*}
where $c_{i}\in \mathbb{F}_{p}^{*}$ for any $0\leq i\leq 4$.
This is impossible since the generator polynomial of $\mathcal{C}$ is $g(x)$ and ${\rm deg}(g(x))=6$.
Hence there does not exist a codeword $c(x)$ in $\mathcal{C}$ with $(w_{H}(c(x)),w_{p}(c(x)))=(5,\,6)$.

%

In order to prove that $\mathcal{C}$ is an AMDS $\left(3p,\,7\right)_p$ symbol-pair code, we need to determine that there does not exist a codeword in $\mathcal{C}$ with $(w_{H}(c(x)),w_{p}(c(x)))=(4,\,6)$ and $(w_{H}(c(x)),w_{p}(c(x)))=(4,\,7)$.

Let $c(x)=\left(x^{3}-1\right)^{t}v(x)$ be a codeword in $\mathcal{C}$ with Hamming weight $4$, where $0\leq t\leq p-1$, $(x^3-1)\nmid v(x)$ and $v(x)=v_{0}(x^3)+x\,v_{1}(x^3)+x^{2}\,v_{2}(x^3)$.
Denote $N_{v}=w_{H}\left(v(x)\,{\rm mod}\,(x^3-1)\right)$.
Then Lemma \ref{lemdH} $(i)$ indicates that $4=\left(1+t\right)N_{v}$.
There are two cases to be considered: $\left(N_{v},\,t\right)=\left(1,\,2\right)$ or $\left(N_{v},\,t\right)=\left(2,\,1\right)$.

Now we first derive that there does not exist a codeword in $\mathcal{C}$ with $(w_{H}(c(x)),w_{p}(c(x)))=(4,\,6)$.
If $\left(N_{v},\,t\right)=\left(1,\,2\right)$, then it is routine to check that the symbol-pair weight of $c(x)$ is greater than $6$.
If $\left(N_{v},\,t\right)=\left(2,\,1\right)$ and $c(x)$ has symbol-pair weight $6$, then
suppose
\begin{equation*}
  c(x)=1+a_{1}\,x+a_{2}\,x^{3i}+a_{3}\,x^{3i+1}
\end{equation*}
for some positive integer $i$ with $1\leq i\leq p-1$ and $a_{1},\,a_{2},\,a_{3}\in \mathbb{F}_{p}^*$.
It follows from $3\,|\,(p-1)$ that $p\nmid 3i$.
The fact $c^{(2)}(1)=c^{(3)}(1)=0$ yields
\begin{equation}\label{eqAMDS3p7p_4-601}
  \left\{
  \begin{array}{l}
  3i\left(3i-1\right)a_{2}+3i\left(3i+1\right)a_{3}=0,\\[2mm]
  3i\left(3i-1\right)\left(3i-2\right)a_{2}+3i\left(3i+1\right)\left(3i-1\right)
  a_{3}=0,\\
  \end{array}
  \right.
\end{equation}
which implies $3i\left(3i+1\right)a_{3}=0$.
If $p\nmid (3i+1)$, then (\ref{eqAMDS3p7p_4-601}) indicates $p\nmid (3i-1)$, a contradiction.
Thus $a_{3}=0$, which is impossible.

For the remaining case, we claim that there exists a codeword in $\mathcal{C}$ with $(w_{H}(c(x)),w_{p}(c(x)))=(4,\,7)$.
If $c(x)$ has symbol-pair weight $7$, then
assume
\begin{equation*}
  c(x)=1+a_1\,x+a_2\,x^i+a_3\,x^j
\end{equation*}
for some integers $i,\,j$ with $3\leq i<j-1< n-2$ and $a_{1},\,a_{2},\,a_{3}\in \mathbb{F}_{p}^*$.
It follows from $c(1)=c^{(1)}(1)=c^{(2)}(1)=c^{(3)}(1)=0$ that
\begin{equation}\label{AMDS3p7peq01}
  \left\{
  \begin{array}{l}
  1+a_{1}+a_{2}+a_{3}=0,\\[2mm]
  a_{1}+i\,a_{2}+j\,a_{3}=0,\\[2mm]
  i\left(i-1\right)\,a_{2}+j\left(j-1\right)\,a_{3}=0,\\[2mm]
  i\left(i-1\right)\left(i-2\right)\,a_{2}
  +j\left(j-1\right)\left(j-2\right)\,a_{3}=0.\\
  \end{array}
  \right.
\end{equation}
Denote by $A$ the corresponding coefficient matrix of $\mathbb{F}_{p}$-linear system of equations (\ref{AMDS3p7peq01}).
Then
\begin{equation*}
  \begin{split}
  A=&\left(
  \begin{array}{cccc}
    1   &1             &1             &-1                         \\
    1   &i             &j             &0                          \\
    0   &i(i-1)        &j(j-1)        &0                          \\
    0   &i(i-1)(i-2)   &j(j-1)(j-2)   &0                          \\
  \end{array}\right)\\[2mm]
  \rightarrow&
  \left(\begin{array}{cccc}
    1 &1    &1                                   &-1                      \\
    0 &i-1  &j-1                                 &1                       \\
    0 &0    &\left(j-i\right)\left(j-1\right)    &-i                      \\
    0 &0    &0                                   &ij                      \\
  \end{array}\right).
  \end{split}
\end{equation*}

Therefore, (\ref{AMDS3p7peq01}) has nonzero solutions if $p\,|\,j$ or $p\,|\,i,\, p\,|\left(j-1\right)$ and there exists a codeword in $\mathcal{C}$ with $(w_{H}(c(x)),w_{p}(c(x)))=(4,\,7)$.

Consequently, $\mathcal{C}$ is an AMDS $\left(3p,\,7\right)_{p}$ symbol-pair code.
The proof is completed.
\end{proof}

In the following, a class of AMDS symbol-pair codes with $n=3p$ and $d_{p}=8$ is proposed.

\begin{thm}\label{thmAMDS3p8p}
Let $p$ be an odd prime with $3\,|\left(p-1\right)$.
Then there exists an AMDS $\left(3p,\,8\right)_{p}$ symbol-pair code.
\end{thm}

\begin{proof}
Let $\mathcal{C}$ be a repeated-root cyclic code of length $3p$ over $\mathbb{F}_p$ with generator polynomial
\begin{equation*}
g(x)=\left(x-1\right)^4\left(x-\omega\right)^2\left(x-\omega^2\right)
\end{equation*}
where $\omega$ is a primitive third root of unity in $\mathbb{F}_{p}$.
It can be checked that $\omega\ne -2$.
By Lemma \ref{lemdistance}, we get that $\mathcal{C}$ is a $\left[3p,\,3p-7,\,5\right]$ cyclic code.
By Lemma \ref{lemdh3}, one can immediately get $d_{p}\left(\mathcal{C}\right)\geq8$
since $\left(x^3-1\right)\,\big|\,g(x)$ and $2<5=d_{H}(\mathcal{C})<3p-(3p-7)=7$.

To derive that $\mathcal{C}$ is an AMDS $\left(3p,\,8\right)_{p}$ symbol-pair code, it suffices to show that there exists a codeword in $\mathcal{C}$ with symbol-pair weight $8$.
Now we show that there exists codeword in $\mathcal{C}$ with $(w_{H}(c(x)),w_{p}(c(x)))=(6,\,8)$.

  For the subcase of $c(x)=1+a_{1}\,x+a_{2}\,x^2+a_{3}\,x^{3}+a_{4}\,x^{3i+1}+a_{5}\,x^{3i+2}$ with $2\leq i\leq p-2$ and $a_{i}\in \mathbb{F}_p^{*}$ for any $1\leq i\leq5$.
  It follows from $c\left(1\right)=c\left(\omega\right)=c\left(\omega^2\right)=c^{(1)}\left(1\right)=c^{(1)}\left(\omega\right)=0$ that
  \begin{equation*}
    a_1=-a_4,\quad a_2=-a_5,\quad a_3=-1,\quad a_4=-\frac{\omega}{i},\quad a_5=\frac{\omega+1}{i}
  \end{equation*}
  due to $p\nmid 3i$.
  Then $c^{(2)}\left(1\right)=c^{(3)}\left(1\right)=0$ induces
  \begin{equation*}
  \left\{
  \begin{array}{l}
    2\,a_{2}+6\,a_{3}+3i\left(3i+1\right)a_{4}+
    \left(3i+2\right)\left(3i+1\right)a_{5}=0,\\[2mm]
    6\,a_{3}+3i\left(3i+1\right)\left(3i-1\right)a_{4}
    +3i\left(3i+1\right)\left(3i+2\right)a_{5}=0\\
  \end{array}
  \right.
  \end{equation*}
  which yields that
  \begin{equation*}
    \omega=-\frac{3i+1}{2}=-\frac{3i\left(i+1\right)}{3i+1}.
  \end{equation*}
  Hence $3i^2+1=0$.
  It follows that in this subcase, there exists a codeword
  \begin{equation*}
    c(x)=1+\frac{3i+3}{3i+1}\,x+\frac{3i-2}{3i+1}\,x^2-x^3
    -\frac{3i+3}{3i+1}\,x^{3i+1}-\frac{3i-2}{3i+1}\,x^{3i+2}
  \end{equation*}
  with Hamming weight $6$ and symbol-pair weight $8$.

Therefore, there exists a codeword in $\mathcal{C}$ with $(w_{H}(c(x)),w_{p}(c(x)))=(6,\,8)$.
As a result, $\mathcal{C}$ is an AMDS $\left(3p,\,8\right)_{p}$ symbol-pair code.
This completes the proof.
\end{proof}

%
%

\begin{thm}\label{thmAMDS3p12p}
Let $p$ be an odd prime with $3\,|\left(p-1\right)$.
Then there exists an AMDS $\left(3p,\,12\right)_{p}$ symbol-pair code.
\end{thm}

\begin{proof}
Let $\mathcal{C}$ be a repeated-root cyclic code of length $3p$ over $\mathbb{F}_p$ with generator polynomial
\begin{equation*}
  g(x)=\left(x-1\right)^{5}\left(x-\omega\right)^3\left(x-\omega^2\right)^3
\end{equation*}
where $\omega$ is a primitive third root of unity in $\mathbb{F}_{p}$.
By Lemma \ref{lemdistance}, one can immediately get $d_{H}(\mathcal{C})=6$.
Since the code in Lemma \ref{lemML} is a subcode of $\mathcal{C}$ and the minimum symbol-pair distance of the code in Lemma \ref{lemML} is $12$, one can derive that $d_p(\mathcal{C})\geq 12$.
Combining with $d_p(\mathcal{C})\leq 2 d_H(\mathcal{C})=12$, one can obtain that $\mathcal{C}$ is an AMDS $\left(3p,\,12\right)_{p}$ symbol-pair code.
\end{proof}

\begin{thm}\label{thmAMDS3p13p}
Let $p$ be an odd prime with $3\,|\left(p-1\right)$.
Then there exists an AMDS $\left(3p,\,13\right)_{p}$ symbol-pair code.
\end{thm}

\begin{proof}
Let $\mathcal{C}$ be a repeated-root cyclic code of length $3p$ over $\mathbb{F}_p$ with generator polynomial
\begin{equation*}
  g(x)=\left(x-1\right)^{6}\left(x-\omega\right)^3\left(x-\omega^2\right)^3
\end{equation*}
where $\omega$ is a primitive third root of unity in $\mathbb{F}_{p}$.
By Lemma \ref{lemdistance}, one can immediately get $d_{H}(\mathcal{C})=7$.
Since the code in Lemma \ref{lemML} is a subcode of $\mathcal{C}$ and the minimum symbol-pair distance of the code in Lemma \ref{lemML} is $12$, one can derive that $d_p(\mathcal{C})\geq 12$.
It can be checked that there does not exist codeword in $\mathcal{C}$ with Hamming weight $7$ and symbol-pair weight $12$.

In the sequel, we determine that there does not exist codeword in $\mathcal{C}$ with $(w_{H}(c(x)),w_{p}(c(x)))=(8,\,12)$, $(w_{H}(c(x)),w_{p}(c(x)))=(9,\,12)$ and $(w_{H}(c(x)),w_{p}(c(x)))=(10,\,12)$.

Firstly, assume that there exists a codeword $c(x)$ in $\mathcal{C}$ with $(w_{H}(c(x)),w_{p}(c(x)))=(8,\,12)$.
There are five subcases to be considered:
\begin{itemize}
  \item
  For the subcase of $c(x)=1+a_{1}\,x+a_{2}\,x^{2}+a_{3}\,x^{3}+a_{4}\,x^{4}+a_{5}\,x^{l_{1}}+a_{6}\,x^{l_{2}}+a_{7}\,x^{l_{3}}$ with $7\leq l_{1}+1<l_{2}<l_{3}-1\leq 3p-3$ and $a_{i}\in \mathbb{F}_p^{*}$ for any $1\leq i\leq7$.
  \begin{itemize}
    \item
    If $\mathcal{CS}\left(l_{1},\,l_{2},\,l_{3}\right)\notin\{(0,\,0,\,2),(0,\,1,\,2),(0,\,2,\,2),(1,\,1,\,2),(1,\,2,\,2)\}$, then it can be checked that by $c\left(1\right)=c\left(\omega\right)=c\left(\omega^2\right)=0$, one can immediately get $a_{2}=0$, a contradiction.

    \item
    If $\mathcal{CS}\left(l_{1},\,l_{2},\,l_{3}\right)\in\{(0,\,0,\,2),(0,\,1,\,2),(0,\,2,\,2),(1,\,1,\,2),(1,\,2,\,2)\}$, then $c\left(1\right)=c\left(\omega\right)=c\left(\omega^2\right)=c^{\left(1\right)}\left(1\right)
    =c^{\left(1\right)}\left(\omega\right)=c^{\left(1\right)}\left(\omega^2\right)=0$ yields that $a_{i}=0$ for some $1\leq i\leq4$, a contradiction.
  \end{itemize}

  \item
  For the subcase of $c(x)=1+a_{1}\,x+a_{2}\,x^{2}+a_{3}\,x^3+a_{4}\,x^{l_{1}}+a_{5}\,x^{l_{1}+1}+a_{6}\,x^{l_{2}}+a_{7}\,x^{l_{3}}$ with $7\leq l_{1}+2<l_{2}<l_{3}-1\leq 3p-3$ and $a_{i}\in \mathbb{F}_p^{*}$ for any $1\leq i\leq7$.
  \begin{itemize}
    \item
    If $\mathcal{V}\left(l_{1},\,l_{2},\,l_{3}\right)\in\{(0,\,0,\,0),\,(0,\,0,\,1),\,(0,\,1,\,0),\,(0,\,1,\,1),\,(2,\,0,\,0),\,(2,\,0,\,2)\}$, then by $c\left(1\right)=c\left(\omega\right)=c\left(\omega^2\right)=0$, one can immediately derive $a_{1}=0$ or $a_{2}=0$, a contradiction.

    \item
    If $\mathcal{V}\left(l_{1},\,l_{2},\,l_{3}\right)\notin\{(0,\,0,\,0),\,(0,\,0,\,1),\,(0,\,1,\,0),\,(0,\,1,\,1),\,(2,\,0,\,0),\,(2,\,0,\,2)\}$, then $c^{\left(1\right)}\left(1\right)=c^{\left(1\right)}\left(\omega\right)=c^{\left(1\right)}\left(\omega^2\right)=0$ implies that $a_{i}=0$ for some $1\leq i\leq7$, a contradiction.
  \end{itemize}

  \item
  For the subcase of $c(x)=1+a_{1}\,x+a_{2}\,x^2+a_{3}\,x^{l_{1}}+a_{4}\,x^{l_{1}+1}+a_{5}\,x^{l_{1}+2}+a_{6}\,x^{l_{2}}+a_{7}\,x^{l_{3}}$ with $7\leq l_{1}+3<l_{2}<l_{3}-1\leq 3p-3$ and $a_{i}\in \mathbb{F}_p^{*}$ for any $1\leq i\leq7$.
  For any $7\leq l_{1}+3<l_{2}<l_{3}-1\leq 3p-3$, $c^{\left(1\right)}\left(1\right)=c^{\left(1\right)}\left(\omega\right)=c^{\left(1\right)}\left(\omega^2\right)=0$ implies that $a_{i}=0$ for some $1\leq i\leq7$, a contradiction.

  \item
  For the subcase of $c(x)=1+a_{1}\,x+a_{2}\,x^2+a_{3}\,x^{l_{1}}+a_{4}\,x^{l_{1}+1}+a_{5}\,x^{l_{2}}+a_{6}\,x^{l_{2}+1}+a_{7}\,x^{l_{3}}$ with $6\leq l_{1}+2<l_{2}<l_{3}-2\leq 3p-4$ and $a_{i}\in \mathbb{F}_p^{*}$ for any $1\leq i\leq7$.
  \begin{itemize}
    \item
    If $\mathcal{V}\left(l_{1},\,l_{2},\,l_{3}\right)\in\{(0,\,0,\,0),\,(0,\,0,\,1),\,(1,\,1,\,1),\,(1,\,1,\,2),\,(2,\,2,\,0),\,(2,\,2,\,2)\}$, then by $c\left(1\right)=c\left(\omega\right)=c\left(\omega^2\right)=0$, one can immediately derive $1=0$, $a_{1}=0$ or $a_{2}=0$.

    \item
    If $\mathcal{V}\left(l_{1},\,l_{2},\,l_{3}\right)\in\{(0,\,0,\,0),\,(0,\,0,\,1),\,(1,\,1,\,1),\,(1,\,1,\,2),\,(2,\,2,\,0),\,(2,\,2,\,2)\}$, then $c^{\left(1\right)}\left(1\right)=c^{\left(1\right)}\left(\omega\right)=c^{\left(1\right)}\left(\omega^2\right)=0$ implies that $a_{i}=0$ for some $1\leq i\leq7$, a contradiction.
  \end{itemize}

  \item
  For the subcase of $c(x)=1+a_{1}\,x+a_{2}\,x^{l_{1}}+a_{3}\,x^{l_{1}+1}+a_{4}\,x^{l_{2}}+a_{5}\,x^{l_{2}+1}+a_{6}\,x^{l_{3}}+a_{7}\,x^{l_{3}+1}$ with $5\leq l_{1}+2<l_{2}<l_{3}-2\leq 3p-5$ and $a_{i}\in \mathbb{F}_p^{*}$ for any $1\leq i\leq7$.
  \begin{itemize}
    \item
    If $\mathcal{CS}\left(l_{1},\,l_{2},\,l_{3}\right)\in\{(0,\,0,\,1),\,(0,\,0,\,2),\,(1,\,1,\,1),\,(2,\,2,\,2)\}$, then by $c\left(1\right)=c\left(\omega\right)=c\left(\omega^2\right)=0$, one can immediately derive $1=0$ or $a_{1}=0$ for some  $1\leq i\leq7$.

    \item
    If $\mathcal{CS}\left(l_{1},\,l_{2},\,l_{3}\right)\in\{(0,\,1,\,1),\,(0,\,1,\,2),\,(0,\,2,\,2),\,(1,\,2,\,2)\}$, then $c^{\left(1\right)}\left(1\right)=c^{\left(1\right)}\left(\omega\right)=c^{\left(1\right)}\left(\omega^2\right)=0$ implies that $a_{i}=0$ for some $1\leq i\leq7$, a contradiction.

    \item
    If $\mathcal{CS}\left(l_{1},\,l_{2},\,l_{3}\right)=(0,\,0,\,0)$, then $c\left(1\right)=c\left(\omega\right)
    =c^{\left(1\right)}\left(1\right)=c^{\left(1\right)}\left(\omega\right)=c^{\left(2\right)}\left(1\right)
    =c^{\left(2\right)}\left(\omega\right)=0$ induces
    \begin{equation}\label{eq_AMDS3p12_8-12_222201}
    \left\{
    \begin{array}{l}
      3i\,a_{2}+3j\,a_{4}+3k\,a_{6}=3i\,a_{3}+3j\,a_{5}+3k\,a_{7}=0,\\[2mm]
      3j(3j-3i)\,a_{4}+3k(3k-3i)\,a_{6}=3j(3j-3i)\,a_{5}+3k(3k-3i)\,a_{7}=0
    \end{array}
    \right.
    \end{equation}
    The fact $c^{\left(3\right)}\left(1\right)=0$ indicates
    \begin{equation*}
    \begin{split}
      0&=3i(3i-1)(3i-2)\,a_{2}+3i(3i-1)(3i+1)\,a_{3}+3j(3j-1)(3j-2)\,a_{4}
      \\&+3j(3j-1)(3j+1)\,a_{5}+3k(3k-1)(3k-2)\,a_{6}+3k(3k-1)(3k+1)\,a_{7}
    \end{split}
    \end{equation*}
    which yields
    \begin{equation*}
      a_{4}+a_{5}=0 \quad {\rm and} \quad a_{6}+a_{7}=0
    \end{equation*}
    due to (\ref{eq_AMDS3p12_8-12_222201}).
    It follows from $c^{\left(4\right)}\left(1\right)=0$ that
    \begin{equation*}
    \begin{split}
      0&=3i(3i-1)(3i-2)(3i-3)\,a_{2}+3i(3i-1)(3i-2)(3i+1)\,a_{3}
      \\&+3j(3j-1)(3j-2)(3j-3)\,a_{4}+3j(3j-1)(3j-2)(3j+1)\,a_{5}
      \\&+3k(3k-1)(3k-2)(3k-3)\,a_{6}+3k(3k-1)(3k-2)(3k+1)\,a_{7}.
    \end{split}
    \end{equation*}
    By (\ref{eq_AMDS3p12_8-12_222201}), one can get
    \begin{equation*}
      3k(3k-3i)(3k-3j)\,a_{6}=0.
    \end{equation*}
    This implies $a_{6}=0$, a contradiction.
  \end{itemize}
\end{itemize}

Secondly, assume that there exists a codeword $c(x)$ in $\mathcal{C}$ with $(w_{H}(c(x)),w_{p}(c(x)))=(9,\,12)$.
There are six subcases to be considered:
\begin{itemize}
  \item
  For the subcase of
  $c(x)=1+a_{1}\,x+a_{2}\,x^{2}+a_{3}\,x^{3}+a_{4}\,x^{4}+a_{5}\,x^{5}+a_{6}\,x^{6}+a_{7}\,x^{l_1}+a_{8}\,x^{l_2}$ with $8\leq l_1< l_2-1\leq 3p-3$ and $a_{i}\in \mathbb{F}_p^{*}$ for any $1\leq i\leq8$.
  \begin{itemize}
    \item
    If $\mathcal{V}(l_1,\,l_2)=\{1,\,2\}$, then $c^{\left(1\right)}\left(1\right)=c^{\left(1\right)}\left(\omega\right)
    =c^{\left(1\right)}\left(\omega^2\right)=c^{\left(2\right)}\left(1\right)
    =c^{\left(2\right)}\left(\omega\right)=c^{\left(2\right)}\left(\omega^2\right)=0$ induces
    \begin{equation*}
    \left\{
    \begin{array}{l}
      3\,a_{3}+6\,a_{6}=0,\\[2mm]
      6\,a_{3}+30\,a_{6}=0.
    \end{array}
    \right.
    \end{equation*}
    which indicates that $a_{6}=0$, a contradiction.

    \item
    If $\mathcal{V}(l_1,\,l_2)\ne \{1,\,2\}$, then by a similar arguments to the previous $\mathcal{V}(l_1,\,l_2)=\{1,\,2\}$, one can obtain a contradiction and we omit it here.

  \end{itemize}

  \item
  For the subcase of $c(x)=1+a_{1}\,x+a_{2}\,x^{2}+a_{3}\,x^{3}+a_{4}\,x^{4}+a_{5}\,x^{5}+a_{6}\,x^{l_1}+a_{7}\,x^{l_1+1}+a_{8}\,x^{l_2}$ with $8\leq l_1+1<l_2-1\leq 3p-3$ and $a_{i}\in \mathbb{F}_p^{*}$ for any $1\leq i\leq8$.
  It can be checked that $a_{i}=0$ for some $1\leq i\leq7$, a contradiction.

  \item
  For the subcase of
  $c(x)=1+a_{1}\,x+a_{2}\,x^{2}+a_{3}\,x^{3}+a_{4}\,x^{4}+a_{5}\,x^{l_1}+a_{6}\,x^{l_1+1}+a_{7}\,x^{l_1+2}+a_{8}\,x^{l_2}$ with $8\leq l_1+2<l_2-1\leq 3p-3$ and $a_{i}\in \mathbb{F}_p^{*}$ for any $1\leq i\leq8$.
  It can be checked that $a_{i}=0$ for some $1\leq i\leq7$, a contradiction.

  \item
  For the subcase of
  $c(x)=1+a_{1}\,x+a_{2}\,x^{2}+a_{3}\,x^{3}+a_{4}\,x^{4}+a_{5}\,x^{l_1}+a_{6}\,x^{l_1+1}+a_{7}\,x^{l_2}+a_{8}\,x^{l_2+1}$ with $7\leq l_1+1<l_2-1\leq 3p-4$ and $a_{i}\in \mathbb{F}_p^{*}$ for any $1\leq i\leq8$.
  It can be checked that $a_{i}=0$ for some $1\leq i\leq7$, a contradiction.

  \item
  For the subcase of
  $c(x)=1+a_{1}\,x+a_{2}\,x^{2}+a_{3}\,x^{3}+a_{4}\,x^{l_1}+a_{5}\,x^{l_1+1}+a_{6}\,x^{l_1+2}+a_{7}\,x^{l_1+3}+a_{8}\,x^{l_2}$ with $8\leq l_1+3<l_2-1\leq 3p-3$ and $a_{i}\in \mathbb{F}_p^{*}$ for any $1\leq i\leq8$.
  It can be checked that $a_{i}=0$ for some $1\leq i\leq7$, a contradiction.

  \item
  For the subcase of
  $c(x)=1+a_{1}\,x+a_{2}\,x^{2}+a_{3}\,x^{3}+a_{4}\,x^{l_1}+a_{5}\,x^{l_1+1}+a_{6}\,x^{l_1+2}+a_{7}\,x^{l_2}+a_{8}\,x^{l_2+1}$ with $7\leq l_1+2<l_2-1\leq 3p-3$ and $a_{i}\in \mathbb{F}_p^{*}$ for any $1\leq i\leq8$.
  It can be checked that $a_{i}=0$ for some $1\leq i\leq7$, a contradiction.
  \end{itemize}

Thirdly, assume that there exists a codeword $c(x)$ in $\mathcal{C}$ with $(w_{H}(c(x)),w_{p}(c(x)))=(10,\,12)$.
There are three subcases to be considered:
\begin{itemize}
  \item
  For the subcase of
  $c(x)=1+a_{1}\,x+a_{2}\,x^{2}+a_{3}\,x^{3}+a_{4}\,x^{4}+a_{5}\,x^{5}+a_{6}\,x^{6}+a_{7}\,x^{7}+a_{8}\,x^{8}+a_{9}\,x^{l}$ with $10\leq l\leq 3p-2$ and $a_{i}\in \mathbb{F}_p^{*}$ for any $1\leq i\leq9$.
  \begin{itemize}
    \item
    If $\mathcal{V}(l)=0$, then $c\left(1\right)=c\left(\omega\right)
    =c\left(\omega^2\right)=c^{\left(1\right)}\left(1\right)
    =c^{\left(1\right)}\left(\omega\right)=c^{\left(1\right)}\left(\omega^2\right)==c^{\left(1\right)}\left(1\right)
    =c^{\left(1\right)}\left(\omega\right)=c^{\left(1\right)}\left(\omega^2\right)=0$ induces
    \begin{equation*}
    \left\{
    \begin{array}{l}
      a_{2}+a_{5}+a_{8}=0,\\[2mm]
      2\,a_{2}+5\,a_{5}+8\,a_{8}=0,\\[2mm]
      60\,a_{5}+336\,a_{8}=0.
    \end{array}
    \right.
    \end{equation*}
    which indicates that $a_{8}=0$, a contradiction.

    \item
    If $\mathcal{V}(l)=1$, then by a similar arguments to the previous $\mathcal{V}(l)=0$, one can obtain a contradiction and we omit it here.

    \item
    If $\mathcal{V}(l)=2$, then  by a similar arguments to the previous $\mathcal{V}(l)=0$, one can obtain a contradiction and we omit it here.
  \end{itemize}

  \item
  For the subcase of $c(x)=1+a_{1}\,x+a_{2}\,x^{2}+a_{3}\,x^{3}+a_{4}\,x^{4}+a_{5}\,x^{5}+a_{6}\,x^{6}+a_{7}\,x^{7}+a_{8}\,x^{l}+a_{9}\,x^{l+1}$ with $9\leq l\leq 3p-3$ and $a_{i}\in \mathbb{F}_p^{*}$ for any $1\leq i\leq9$.
  It can be checked that no matter $\mathcal{V}(l)=0$, $\mathcal{V}(l)=1$ or $\mathcal{V}(l)=2$, there exists a contradiction.

  \item
  For the subcase of
  $c(x)=1+a_{1}\,x+a_{2}\,x^{2}+a_{3}\,x^{3}+a_{4}\,x^{4}+a_{5}\,x^{5}+a_{6}\,x^{6}+a_{7}\,x^{l}+a_{8}\,x^{l+1}+a_{9}\,x^{l+2}$ with $8\leq l\leq 3p-4$ and $a_{i}\in \mathbb{F}_p^{*}$ for any $1\leq i\leq9$.
  It can be checked that no matter $\mathcal{V}(l)=0$, $\mathcal{V}(l)=1$ or $\mathcal{V}(l)=2$, one can obtain that $a_i=0$ for some $1\leq i\leq 9$, a contradiction.

  \item
  For the subcase of
  $c(x)=1+a_{1}\,x+a_{2}\,x^{2}+a_{3}\,x^{3}+a_{4}\,x^{4}+a_{5}\,x^{5}+a_{6}\,x^{l}+a_{7}\,x^{l+1}+a_{8}\,x^{l+2}+a_{9}\,x^{l+3}$ with $7\leq l\leq 3p-5$ and $a_{i}\in \mathbb{F}_p^{*}$ for any $1\leq i\leq9$.
  It can be checked that no matter $\mathcal{V}(l)=0$, $\mathcal{V}(l)=1$ or $\mathcal{V}(l)=2$, one can obtain that $a_i=0$ for some $1\leq i\leq 9$, a contradiction.

  \item
  For the subcase of
  $c(x)=1+a_{1}\,x+a_{2}\,x^{2}+a_{3}\,x^{3}+a_{4}\,x^{4}+a_{5}\,x^{l}+a_{6}\,x^{l+1}+a_{7}\,x^{l+2}+a_{8}\,x^{l+3}+a_{9}\,x^{l+4}$ with $8\leq l\leq 3p-4$ and $a_{i}\in \mathbb{F}_p^{*}$ for any $1\leq i\leq9$.
  It can be checked that no matter $\mathcal{V}(l)=0$, $\mathcal{V}(l)=1$ or $\mathcal{V}(l)=2$, one can obtain that $a_i=0$ for some $1\leq i\leq 9$, a contradiction.
  \end{itemize}

To prove that $\mathcal{C}$ is an AMDS $\left(3p,\,13\right)_{p}$ symbol-pair code, it suffices to determine that there exists a codeword in $\mathcal{C}$ with symbol-pair weight $13$.
Now we prove that there exist codeword in $\mathcal{C}$ with $(w_H(c(x)),\,w_p(c(x)))=(8,\,13)$.

For the subcase of $c(x)=1+a_{1}\,x+a_{2}\,x^{l_1}+a_{3}\,x^{l_1+1}+a_{4}\,x^{l_2}+a_{5}\,x^{l_2+1}+a_{6}\,x^{l_3}+a_{7}\,x^{l_4}$ with $7\leq l_1+4<l_2+2<l_3<l_4-1\leq 3p-3$, $\mathcal{V}\left(l_{1},\,l_{2},\,l_{3},\,l_{4}\right)=(0,\,0,\,0,\,1)$ and $a_{i}\in \mathbb{F}_p^{*}$ for any $1\leq i\leq7$.

It follows from $c(1)=c(\omega)=c(\omega^2)=0$ that
\begin{equation*}
  \left\{
  \begin{array}{l}
  1+a_{1}+a_{2}+a_{3}+a_{4}+a_{5}+a_{6}+a_{7}=0,\\[2mm]
  1+a_{1}\,\omega+a_{2}+a_{3}\,\omega+a_{4}+a_{5}\,\omega+a_{6}+a_{7}\,\omega=0,\\[2mm]
  1+a_{1}\,\omega^2+a_{2}+a_{3}\,\omega^2+a_{4}+a_{5}\,\omega^2+a_{6}+a_{7}\,\omega^2=0\\
  \end{array}
  \right.
\end{equation*}
which leads to
\begin{equation}\label{AMDS3p13eq01}
  \left\{
  \begin{array}{l}
  1+a_{2}+a_{4}+a_{6}=0,\\[2mm]
  a_{1}+a_{3}+a_{5}+a_{7}=0.\\
  \end{array}
  \right.
\end{equation}
It can be checked that $a_1=a_4$, $a_2=a_3$, $a_6=a_7$, $a_5=1$, $l_2=2\,l_1$, $l_4=l_3+4$ is a solution of the equations (\ref{AMDS3p13eq01}).
By $c^{(1)}(1)=c^{(1)}(\omega)=c^{(1)}(\omega^2)=0$, one can obtain that
\begin{equation*}
  \left\{
  \begin{array}{l}
  a_{1}+l_1\,a_{2}+\left(l_1+1\right)a_{3}+l_2\,a_{4}+\left(l_2+1\right)a_{5}+l_3\,a_{6}+l_4\,a_{7}=0,\\[2mm]
  a_{1}+l_1\,a_{2}\,\omega^2+\left(l_1+1\right)a_{3}+l_2\,a_{4}\,\omega^2+\left(l_2+1\right)a_{5}
  +l_3\,a_{6}\,\omega^2+l_4\,a_{7}=0,\\[2mm]
  a_{1}+l_1\,a_{2}\,\omega+\left(l_1+1\right)a_{3}+l_2\,a_{4}\,\omega+\left(l_2+1\right)a_{5}
  +l_3\,a_{6}\,\omega+l_4\,a_{7}=0\\
  \end{array}
  \right.
\end{equation*}
which yields
\begin{equation}\label{AMDS3p13eq02}
  \left\{
  \begin{array}{l}
  a_{1}+\left(l_1+1\right)a_{3}+\left(l_2+1\right)a_{5}+l_4\,a_{7}=0,\\[2mm]
  l_1\,a_{2}+l_2\,a_{4}+l_3\,a_{6}=0.\\
  \end{array}
  \right.
\end{equation}
By $a_1=a_4$, $a_2=a_3$, $a_6=a_7$, $a_5=1$, $l_2=2\,l_1$, $l_4=l_3+4$ and (\ref{AMDS3p13eq02}), one can get
\begin{equation}\label{AMDS3p13eq03}
  l_2\left(a_1-1\right)=3\,a_6.
\end{equation}
It follows from $c^{(2)}(1)=c^{(2)}(\omega)=c^{(2)}(\omega^2)=0$ that
\begin{equation*}
  \left\{
  \begin{array}{l}
  l_1\left(l_1-1\right)a_{2}+l_1\left(l_1+1\right)a_{3}+l_2\left(l_2-1\right)a_{4}+l_2\left(l_2+1\right)a_{5}
  +l_3\left(l_3-1\right)a_{6}+l_4\left(l_4-1\right)a_{7}=0,\\[2mm]
  l_1\left(l_1-1\right)a_{2}\,\omega+l_1\left(l_1+1\right)a_{3}\,\omega^2+l_2\left(l_2-1\right)a_{4}\,\omega
  +l_2\left(l_2+1\right)a_{5}\,\omega^2+l_3\left(l_3-1\right)a_{6}\,\omega\\
  +l_4\left(l_4-1\right)a_{7}\,\omega^2=0,\\[2mm]
  l_1\left(l_1-1\right)a_{2}\,\omega^2+l_1\left(l_1+1\right)a_{3}\,\omega+l_2\left(l_2-1\right)a_{4}\,\omega^2
  +l_2\left(l_2+1\right)a_{5}\,\omega+l_3\left(l_3-1\right)a_{6}\,\omega^2\\
  +l_4\left(l_4-1\right)a_{7}\,\omega=0\\
  \end{array}
  \right.
\end{equation*}
which leads to
\begin{equation}\label{AMDS3p13eq04}
  \left\{
  \begin{array}{l}
  l_1\left(l_1-1\right)a_{2}+l_2\left(l_2-1\right)a_{4}+l_3\left(l_3-1\right)a_{6}=0,\\[2mm]
  l_1\left(l_1+1\right)a_{3}+l_2\left(l_2+1\right)a_{5}+l_4\left(l_4-1\right)a_{7}=0.\\
  \end{array}
  \right.
\end{equation}
By $a_1=a_4$, $a_2=a_3$, $a_6=a_7$, $a_5=1$, $l_2=2\,l_1$, $l_4=l_3+4$ and (\ref{AMDS3p13eq01}), one can derive
\begin{equation*}
  \left(2\,l_3-2\,l_1+3\right)a_6=0
\end{equation*}
which leads to
\begin{equation}\label{AMDS3p13eq04}
  2\,l_3+3=2\,l_1
\end{equation}
since $a_6\ne 0$.
It follows from $c^{(3)}(1)=0$ that
\begin{equation*}
  \begin{split}
  &l_1\left(l_1-1\right)\left(l_1-2\right)a_{2}+l_1\left(l_1+1\right)\left(l_1-1\right)a_{3}
  +l_2\left(l_2-1\right)\left(l_2-2\right)a_{4}
  \\&
  +l_2\left(l_2+1\right)\left(l_2-1\right)a_{5}+l_3\left(l_3-1\right)\left(l_3-2\right)a_{6}
  +l_4\left(l_4-1\right)\left(l_4-2\right)a_{7}=0.
  \end{split}
\end{equation*}
By $c^{(4)}(1)=0$, one can get
\begin{equation*}
  \begin{split}
  &l_1\left(l_1-1\right)\left(l_1-2\right)\left(l_1-3\right)a_{2}+l_1\left(l_1+1\right)\left(l_1-1\right)\left(l_1-2\right)a_{3}
  +l_2\left(l_2-1\right)\left(l_2-2\right)\left(l_2-3\right)a_{4}
  \\&
  +l_2\left(l_2+1\right)\left(l_2-1\right)\left(l_2-2\right)a_{5}
  +l_3\left(l_3-1\right)\left(l_3-2\right)\left(l_3-3\right)a_{6}
  +l_4\left(l_4-1\right)\left(l_4-2\right)\left(l_4-3\right)a_{7}=0.
  \end{split}
\end{equation*}
The fact $c^{(5)}(1)=0$ yields that
\begin{equation*}
  \begin{split}
  &l_1\left(l_1-1\right)\left(l_1-2\right)\left(l_1-3\right)\left(l_1-4\right)a_{2}
  +l_1\left(l_1+1\right)\left(l_1-1\right)\left(l_1-2\right)\left(l_1-3\right)a_{3}
  \\&
  +l_2\left(l_2-1\right)\left(l_2-2\right)\left(l_2-3\right)\left(l_2-4\right)a_{4}
  +l_2\left(l_2+1\right)\left(l_2-1\right)\left(l_2-2\right)\left(l_2-3\right)a_{5}
  \\&
  +l_3\left(l_3-1\right)\left(l_3-2\right)\left(l_3-3\right)\left(l_3-4\right)a_{6}
  +l_4\left(l_4-1\right)\left(l_4-2\right)\left(l_4-3\right)\left(l_4-4\right)a_{7}=0.
  \end{split}
\end{equation*}
Hence, there exists a codeword in $\mathcal{C}$ with symbol-pair weight $13$.

Therefore, $\mathcal{C}$ is an AMDS $\left(3p,\,13\right)_{p}$ symbol-pair code.
This proves the desired conclusion.
\end{proof}

In the sequel, we give several examples to illustrate the results in Theorems \ref{thmAMDSlps4q}-\ref{thmAMDS3p13p}.

\begin{example}
(1) Let $\mathcal{C}$ be a repeated-root cyclic code of length $12$ over $\mathbb{F}_3$ with generator polynomial
\begin{equation*}
  g(x)=\left(x-1\right)\left(x-\omega^2\right)\left(x-\omega^{6}\right)
\end{equation*}
where $\omega$ is a primitive element of $\mathbb{F}_{9}$ with $\omega^2+2\omega+2=0$.
It can be verified by MAGMA that $\mathcal{C}$ is a $\left[12,\,9,\,2\right]$ code and the minimum pair-distance of $\mathcal{C}$ is $4$, which coincides with Theorem \ref{thmAMDSlps4q} by taking $l=4$, $s=1$ and $q=3$.

(2) Let $\mathcal{C}$ be a repeated-root cyclic code of length $12$ over $\mathbb{F}_3$ with generator polynomial
\begin{equation*}
  g(x)=\left(x-1\right)\left(x^4-1\right).
\end{equation*}
Then it can be verified by MAGMA that $\mathcal{C}$ is a $\left[12,\,7,\,3\right]$ code and the minimum pair-distance of $\mathcal{C}$ is $6$, which coincides with Theorem \ref{thmAMDS4p6p}.

(3) Let $\mathcal{C}$ be a repeated-root cyclic code of length $15$ over $\mathbb{F}_5$ with generator polynomial
\begin{equation*}
  g(x)=(x-1)^{4}\left(x^2+x+1\right).
\end{equation*}
It can be checked by MAGMA that $\mathcal{C}$ is a $\left[15,\,9,\,4\right]$ code and the minimum pair-distance of $\mathcal{C}$ is $7$, which is consistent with Theorem \ref{thmAMDS3p7p}.

(4) Let $\mathcal{C}$ be a repeated-root cyclic code of length $21$ over $\mathbb{F}_7$ with generator polynomial
\begin{equation*}
  g(x)=(x-1)^{4}\left(x-2\right)^2\left(x-4\right).
\end{equation*}
Then it can be verified by MAGMA that $\mathcal{C}$ is a $\left[21,\,14,\,5\right]$ code and the minimum pair-distance of $\mathcal{C}$ is $8$, which coincides with Theorem \ref{thmAMDS3p8p}.

(5) Let $\mathcal{C}$ be a repeated-root cyclic code of length $21$ over $\mathbb{F}_{7}$ with generator polynomial
\begin{equation*}
  g(x)=\left(x-1\right)^{5}\left(x-2\right)^{3}\left(x-2^2\right)^{3}.
\end{equation*}
It can be verified by MAGMA that $\mathcal{C}$ is a $[21,\,10,\,6]$ code and the minimum symbol-pair distance of $\mathcal{C}$ is $12$, which is consistent with Theorem \ref{thmAMDS3p12p}.

(6) Let $\mathcal{C}$ be a repeated-root cyclic code of length $21$ over $\mathbb{F}_{7}$ with generator polynomial
\begin{equation*}
  g(x)=\left(x-1\right)^{6}\left(x-2\right)^{3}\left(x-2^2\right)^{3}.
\end{equation*}
It can be verified by MAGMA that $\mathcal{C}$ is a $[21,\,9,\,7]$ code and the minimum symbol-pair distance of $\mathcal{C}$ is $13$, which is consistent with Theorem \ref{thmAMDS3p13p}.
\end{example}

\section{Conclusions}

In this paper, we propose six classes of AMDS symbol-pair codes over $\mathbb{F}_{p}$ with $p$ an odd prime by employing repeated-root cyclic codes.
\begin{itemize}
  \item $[lp^{s},\,lp^{s}-3,\,2]$ code with $d_{p}=4$;
  \item $[4p,\,4p-5,\,3]$ code with $d_{p}=6$;
  \item $[3p,\,3p-6,\,4]$ code with $d_{p}=7$;
  \item $[3p,\,3p-7,\,5]$ code with $d_{p}=8$;
  \item $[3p,\,3p-10,\,6]$ code with $d_{p}=12$.
  \item $[3p,\,3p-11,\,7]$ code with $d_{p}=13$.
\end{itemize}
Remarkably, among these codes, a class of AMDS symbol-pair codes with unbounded lengths is derived and the minimum symbol-pair distance can reach $13$.

\end{document}